  \RequirePackage{fix-cm}
  \documentclass[smallextended]{svjour3}       
  \smartqed  

  \usepackage{amssymb}

\usepackage{amsmath,bm}
\usepackage{amsthm}
\usepackage{graphicx}
\usepackage{multirow}
\usepackage{array}
\usepackage{algorithm}
\usepackage{algorithmic}
\usepackage{subfig}
\usepackage{caption}
\usepackage{float}
\usepackage{diagbox}

\newcommand{\PreserveBackslash}[1]{\let\temp=\\#1\let\\=\temp}

\newcolumntype{C}[1]{>{\PreserveBackslash\centering}p{#1}}

\newcommand\nc{\newcommand}
\nc{\eqn}[2]{\begin{eqnarray}\label{eqn:#1} #2 \end{eqnarray}}
\nc{\refe}[1]{(\ref{eqn:#1})}
\begin{document}

  \title{Upper and Lower Bounds on Approximating Weighted Mixed Domination
  }


  \author{Mingyu Xiao 
  }


  \institute{
             Mingyu Xiao \at
             School of Computer Science and Engineering,\\
University of Electronic Science and Technology of China,\\
 Chengdu, China\\
                \email{myxiao@gmail.com}
  }

  \date{Received: date / Accepted: date}

  \maketitle

  \begin{abstract}
  A mixed dominating set of a graph $G = (V, E)$ is a mixed set $D$ of vertices and edges, such that for every edge or vertex, if it is not in $D$, then it is adjacent or incident to at least one vertex or edge in $D$. The mixed domination problem is to find a mixed dominating set with a minimum cardinality. It has applications in system control and some other scenarios and it is $NP$-hard to compute an optimal solution.
This paper studies approximation algorithms and hardness of the weighted mixed dominating set problem. The weighted version is a generalization of the unweighted version, where all vertices are assigned the same nonnegative weight $w_v$ and all edges are assigned the same nonnegative weight $w_e$, and the question is to find a mixed dominating set with a minimum total weight. Although the mixed dominating set problem has a simple 2-approximation algorithm, few approximation results for the weighted version are known. The main contributions of this paper include:
\begin{enumerate}
\item[1.] for $w_e\geq w_v$, a 2-approximation algorithm;
\item[2.] for $w_e\geq 2w_v$, inapproximability within ratio 1.3606 unless $P=NP$ and within ratio 2 under UGC;
\item[3.] for $2w_v > w_e\geq w_v$, inapproximability within ratio 1.1803 unless $P=NP$ and within ratio 1.5 under UGC;
\item[4.] for $w_e< w_v$, inapproximability within ratio $(1-\epsilon)\ln |V|$ unless $P=NP$ for any $\epsilon >0$.
\end{enumerate}

\keywords{Approximation algorithms\and Inapproximability\and  Domination}
  \end{abstract}

\section{Introduction}
Domination is an important concept in graph theory.
In a graph, a vertex \emph{dominates} itself and all neighbors of it,
and an edge \emph{dominates} itself and all edges sharing an endpoint with it.
The \textsc{Vertex Dominating Set} problem \cite{Hedetniemi1991Bibliography} (resp., \textsc{Edge Dominating Set} problem \cite{Yannakakis1980Edge}) is to find a minimum set of vertices to dominate all vertices (resp., a minimum set of edges to dominate all edges) in a graph.
These two domination problems have many applications in different fields.
For example, in a network, structures like dominating sets play an important role in global flooding to alleviate the so-called broadcast storm problem. A message broadcast only in the dominating set is an efficient way to ensure that it is received by all transmitters in the network, both in terms of energy and interference~\cite{Nieberg2005A}.
More applications and introduction to domination problems can be found in the literature~\cite{survey}.

Domination problems are rich problems in the field of algorithms.
Both \textsc{Vertex Dominating Set} and \textsc{Edge Dominating Set} are $NP$-hard~\cite{Garey1979Computers,Yannakakis1980Edge}. There are several interesting algorithmic results
about the polynomial solvability on special graph \cite{Zhao2011The,Lan2013On}, approximation algorithms \cite{Johnson1973Approximation,Fujito2002A,EMPX}, parameterized algorithms \cite{Xiao2011New,Xiao3eds} and so on.

In this paper, we consider a related domination problem, called the \textsc{Mixed Domination} problem.
Mixed domination is a mixture concept of vertex domination and edge domination, and \textsc{Mixed Domination} requires to find a set of edges and vertices with the minimum cardinality to dominate other edges and vertices in a graph.
\textsc{Mixed Domination} was first proposed by Alavi et al. based on some specific application scenarios and it was named as the \textsc{Total Covering} problem initially~\cite{Alavi1977Total}. Although we prefer to call this problem a ``domination problem'' at present, it has some properties of ``covering problems'' and can also be treated as a kind of covering problems.
For applications of \textsc{Mixed Domination}, a direct application in system control
was introduced by Zhao et al.~\cite{Zhao2011The}. They used it to minimize the number of phase measurement units (PMUs) needed to be placed and maintain the ability of monitoring the entire system.
We can see that \textsc{Mixed Domination} has drawn certain attention since its introduction~\cite{Lan2013On,Manlove1999On,Zhang1992On,Zhao2011The,MDparameterized}.

\textsc{Mixed Domination} is $NP$-hard even on bipartite and chordal graphs and planar bipartite graphs of maximum degree 4~\cite{Manlove1999On}.
Most of known algorithmic results of \textsc{Mixed Domination} are about the polynomial-time solvable cases on special graphs. Zhao et al.~\cite{Zhao2011The} showed that this problem in trees can be solved in polynomial time. Lan et al.~\cite{Lan2013On}
  provided a linear-time algorithm for \textsc{Mixed Domination} in cacti, and introduced a labeling algorithm based on the primal-dual approach for \textsc{Mixed Domination} in trees.
Recently, \textsc{Mixed Domination} was studied from the parameterized perspective~\cite{MDparameterized}. Several parameterized complexity results under different parameters have been proved.

In terms of approximation algorithms, domination problems have also been extensively studied.
It is easy to observe that a maximum matching in a graph is a 2-approximation solution to \textsc{Edge Dominating Set}. But for \textsc{Vertex Dominating Set}, the best known approximation ratio is $\log|V|+1$~\cite{Johnson1973Approximation}. As a combination of \textsc{Edge Dominating Set} and \textsc{Vertex Dominating Set}, \textsc{Mixed Domination} has a simple 2-approximation algorithm~\cite{Hatami2007An}.

We will study approximation algorithms for weighted mixed domination problems.
A mixed dominating set contains both edges and vertices. \textsc{Mixed Domination} does not distinguish them in the solution set, and only considers the cardinality. However, edge and vertex are two different elements and they may have
different contributions or prices in practice. In the application example in~\cite{Zhao2011The}, we select vertices and edges to place phase measurement units (PMUs) on them
to monitor their mixed neighbors' state variables
in an electric power system. The price to place PMUs on edges and vertices may be different due to the different physical structures. It is reasonable to distinguish edge and vertex by setting different weights to them.
So we introduce the following weighted version problem.

\noindent\rule{\linewidth}{0.2mm}
\textsc{Weighted Mixed Domination} (WMD)\\
\textbf{Instance:} A single undirected graph $G=(V,E)$, and two nonnegative values $w_v$ and $w_e$.\\
\textbf{Question:} To find a vertex subset $V_D\subseteq V$ and an edge subset $E_D\subseteq E$ such that \\
(i) any vertex in $V\setminus V_D$ is either an endpoint of an edge in $E_D$ or adjacent to a vertex in $V_D$;\\
(ii) any edge in $E\setminus E_D$ has at least one endpoint that is either an endpoint of an edge in $E_D$ or a vertex in $V_D$;\\
(iii) the value $w_v|V_D|+w_e|E_D|$ is minimized under the above constraints.\\
\rule{\linewidth}{0.2mm}

In \textsc{Weighted Mixed Domination}, all vertices (resp., edges) receive the same weight. Although the weight function may not be very general, the hardness of the problem increases dramatically, especially in approximation algorithms. It is easy to see that the 2-approximation algorithm for the unweighted version in \cite{Hatami2007An} cannot be extended to the weighted version.
In fact, for most domination problems, the weight version may become much harder. For example,
it is trivial to obtain a 2-approximation algorithm for \textsc{Edge Dominating Set}. But for the weighted version of \textsc{Edge Dominating Set},
it took years to achieve the same approximation ratio~\cite{Fujito2002A}.
In order to obtain more tractability results for \textsc{Weighted Mixed Domination}, we consider two cases:
\textsc{Vertex-Favorable Mixed Domination} (VFMD) and \textsc{Edge-Favorable Mixed Domination} (EFMD).
If we add one more requirement $w_v\leq w_e$ in \textsc{Weighted Mixed Domination}, then it becomes \textsc{Vertex-Favorable Mixed Domination}.  \textsc{Edge-Favorable Mixed Domination} is defined in a similar way by adding a requirement $w_e\leq w_v$. In fact, we will further distinguish two cases of \textsc{Vertex-Favorable Mixed Domination} to study its complexity.
We summarize our main algorithmic and complexity results for \textsc{Weighted Mixed Domination} in Table~\ref{table1}, where $\varepsilon$ is any value $>0$.

\begin{table}[h]
\centering
\caption{Upper and lower bounds on approximating WMD}\label{table1}
\begin{center}
\begin{tabular}{c|c|c|c c}\hline
\multicolumn{2}{c|}{\multirow{2}{*}{Problems}} &  \multicolumn{3}{c}{Approximation ratio}\\
\cline{3-5}
\multicolumn{2}{c|}{\multirow{3}{*}{}} & Upper bounds & \multicolumn{2}{c}{Lower bounds} \\
\hline
\multirow{4}{*}{VFMD} & \multirow{2}{*}{$2w_v \leq w_e$}  & \multirow{4}{*}{2}  & $10\sqrt{5}-21-\varepsilon$ & if $P\neq NP$ (Theorem~\ref{th_3}) \\
\multirow{4}{*}{} & \multirow{2}{*}{}  & \multirow{4}{*}{(Theorems~\ref{th_2} and~\ref{th_2-1})}  & $2-\varepsilon$ & under UGC (Theorem~\ref{th_3})\\
\cline{2-2}
\cline{4-5}
\multirow{4}{*}{} & \multirow{2}{*}{$w_v\leq w_e< 2w_v$}  & \multirow{4}{*}{}  & $5\sqrt{5}-10-\varepsilon$ & if $P\neq NP$ (Theorem~\ref{hard22}) \\
\multirow{4}{*}{} & \multirow{2}{*}{}  & \multirow{4}{*}{}  & $1.5-\varepsilon$ & under UGC (Theorem~\ref{hard22})\\
\hline
EFMD & $w_v> w_e$  & --  & $(1-\varepsilon)\ln n$ & if $P\neq NP$ (Theorem~\ref{th_4})\\
\hline
\end{tabular}
\end{center}
\end{table}

This paper is organized as follows. Sections~\ref{sec:prelimi} and~\ref{sec:properties} introduce some basic notations and properties.
Section~\ref{sec:vertex} deals with \textsc{Vertex-Favorable Mixed Domination}. The results for the case that $2w_v \leq w_e$ are obtained by proving its equivalence to the \textsc{Vertex Cover} problem. The case that $w_v\leq w_e< 2w_v$ is harder. Our 2-approximation algorithm is based on a linear programming for  \textsc{Vertex Cover}.
The lower bounds are obtained by a nontrivial reduction from  \textsc{Vertex Cover}. Section~\ref{sec:edge} proves lower bounds for  \textsc{Edge-Favorable Mixed Domination}
based on a reduction from the \textsc{Set Cover} problem. Finally, some concluding remarks are given in Section~\ref{sec:con}.

\section{Preliminaries}\label{sec:prelimi}

In this paper, a graph $G=(V,E)$ stands for an undirected simple graph
with a vertex set $V$ and an edge set $E$.
We use $n=|V|$ and $m=|E|$ to denote the sizes of the vertex set and edge set, respectively.
Let $X$ be a subset of $V$.
We use $G-X$ to denote
the graph obtained from $G$ by removing vertices in $X$
together with all edges incident to vertices in $X$.
Let $G[X]$ denote the graph induced by  $X$,
i.e., $G[X]= G -(V\setminus X)$. For a subgraph or an edge set $G'$, we use $V(G')$ to denote the set of vertices in $G'$.

In a graph, a vertex \emph{dominates} itself, all of its neighbors and all edges taking it as one endpoint;
an edge \emph{dominates} itself, the two endpoints of it and all other edges having a common endpoint.
A mixed set of vertices and edges $D\subseteq V\cup E$ is called a \emph{mixed dominating set},
if any vertex and edge are dominated by at least one element in $D$.
For a mixed set $D$ of vertices and edges, a vertex (resp., edge) in $D$ is called a \emph{vertex element} (resp., \emph{edge element}) of $D$,
and the set of vertex elements (resp., edge elements) may be denoted by $V_D$ (resp., $E_D$).Thus $V_D=V(G)\cap D$.
The set of vertices that appear in any form in $D$ is denoted by $V(D)$, i.e.,
$V(D)=\{v\in V(G)|v\in D~\text{or}~v~\text{is adjacent to an edge in $D$}\}$.
It holds that $V_D\subseteq V(D)$.
\textsc{Mixed Domination} is to find a mixed dominating set of the minimum cardinality, and
\textsc{Weighted Mixed Domination} is to find a mixed dominating set $D$ such that $w_v|V_D|+w_e|E_D|$ is minimized.
A \emph{weighted instance} is a graph with each vertex assigned the same nonnegative weight $w_v$ and each edge assigned the same
nonnegative weight $w_e$.
In a weighted instance, for a mixed set $D$ of vertices and edges (it may only contain vertices or edges), we define
 $w(D)=w_v|D\cap V|+w_e|D\cap E|$.


A vertex set in a graph is called a \emph{vertex cover} if any edge has at least one endpoint in this set and
a vertex set is called an \emph{independent set} if any pair of vertices in it are not adjacent in the graph.
The \textsc{Vertex Cover} problem is to find a vertex cover of the minimum cardinality.
We may use $S_{md}$, $S_{wmd}$ and $S_{vc}$ to denote an optimal solution to \textsc{Mixed Domination}, \textsc{Weighted Mixed Domination} and \textsc{Vertex Cover}, respectively.

\section{Properties}\label{sec:properties}
We introduce some basic properties of \textsc{Mixed Domination} and \textsc{Weighted Mixed Domination} in this section.

\begin{lemma}\label{t_2}
Any mixed dominating set of a graph contains all isolating vertices (i.e. the vertices of degree $0$) as vertex elements.
\end{lemma}

This lemma follows from the definition of mixed dominating sets directly.
Based on this lemma, we can simply include all isolating vertices in the graph to the solution set and assume the graph has no
isolating vertices.
We have said that \textsc{Mixed Domination} is also related to covering problems.
Next, we reveal some relations between  \textsc{Mixed Domination} and \textsc{Vertex Cover}.
By the definitions of vertex covers and mixed dominated sets, we get

\begin{lemma}\label{t_3}
In a graph without isolating vertices, any vertex cover is a mixed dominating set.
\end{lemma}

Recall that for a mixed dominating set $D$, we use $V(D)$ to denote the set of vertices appearing in $D$. On the other hand, we have that
\begin{lemma}\label{t_4}
For any mixed dominating set $D$, the vertex set $V(D)$ is a vertex cover.
\end{lemma}

Recall that $S_{wmd}$ and $S_{vc}$ denote an optimal solution to \textsc{Weighted Mixed Domination} and \textsc{Vertex Cover} respectively.
It is easy to get the following results from above lemmas.

\begin{corollary}\label{c_1}
For any mixed dominating set $D$, it holds that
$$2|D|\geq |V_D|+2|E_D|\geq |S_{vc}|.$$
\end{corollary}

\begin{lemma}\label{t_5}
Let $G$ be an instance of \textsc{Vertex-Favorable Mixed Domination} having no isolating vertices. For any mixed dominating set $D$ and vertex cover $C$ in $G$, it holds that
\begin{center}
$w(S_{wmd})\leq w(C)$ ~and~ $w(S_{vc})\leq 2w(D)$.
\end{center}
\end{lemma}
\begin{proof}
The first inequality follows from Lemma~\ref{t_3} directly.
By Corollary~\ref{c_1} and $w_v\leq w_e$, we have that
$w(S_{vc})=w_v |S_{vc}|\leq 2w_v|D|=2w_v|V_D|+2w_v|E_D|\leq 2w_v|V_D|+2w_e|E_D|=2w(D)$.
\end{proof}

\begin{corollary}\label{c_2}
Let $G$ be an instance of \textsc{Vertex-Favorable Mixed Domination} having no isolating vertices. It holds that
$$w(S_{wmd})\leq w(S_{vc})\leq 2w(S_{wmd}).$$
\end{corollary}

Lemma~\ref{t_5} and Corollary~\ref{c_2} imply the following result.

\begin{theorem}\label{t_6}
For any $\alpha\geq 1$, given an $\alpha$-approximation solution to \textsc{Vertex Cover}, a 2$\alpha$-approximation solution
to \textsc{Vertex-Favorable Mixed Domination} on the same graph can be constructed in linear time.
\end{theorem}
\begin{proof}
For a weighted instance $G$, let $I$ be the set of degree-0 vertices in it.
Let $G'=G-I$. Let $C$ be an $\alpha$-approximate solution to  \textsc{Vertex Cover} in $G$, which is also an $\alpha$-approximate solution to  \textsc{Vertex Cover} in $G'$.
Let $S'_{vc}$ be a minimum  vertex cover in $G'$, and $S'_{wmd}$ be an optimal solution to \textsc{Weighted Mixed Domination} in $G'$.
We will show that $C\cup I$ is a 2$\alpha$-approximation solution
to \textsc{Vertex-Favorable Mixed Domination} in $G$. By Lemmas~\ref{t_2} and~\ref{t_3}, we know that $C\cup I$ is a mixed dominating set in $G$.
By Corollary~\ref{c_2}, we know that
$$w(C)\leq \alpha w(S'_{vc})\leq 2\alpha w(S'_{wmd}).$$
In $G$, the set $S_{wmd}=S'_{wmd}\cup I$ is an optimal solution to \textsc{Weighted Mixed Domination}.
We have
$$w(C)+w(I)\leq 2\alpha w(S'_{wmd})+w(I)\leq 2\alpha(w(S'_{wmd})+w(I))=2\alpha w(S_{wmd}),$$
which implies that $C\cup I$ is a 2$\alpha$-approximation solution
to \textsc{Vertex-Favorable Mixed Domination} in $G$. Furthermore, the set $I$ can be computed in linear time.
\end{proof}

\textsc{Vertex Cover} allows 2-approximation algorithms and then we have that

\begin{corollary}\label{c_3}
\textsc{Vertex-Favorable Mixed Domination} allows polynomial-time 4-approximation algorithms.
\end{corollary}

\section{\textsc{Vertex-Favorable Mixed Domination}}\label{sec:vertex}
We have obtained a simple 4-approximation algorithm for \textsc{Vertex-Favorable Mixed Domination}. In this section, we improve the ratio to 2 and also show some lower bounds.
We will distinguish two cases to study it: $2w_v \leq w_e$; $w_v\leq w_e< 2w_v$.

\subsection{The case that $2w_v \leq w_e$}
This is the easier case. In fact, we will reduce this case to \textsc{Vertex Cover} and also reduce
\textsc{Vertex Cover} to it, keeping the approximation ratio. Thus, for this case we will get the same approximation upper and lower bounds as that of \textsc{Vertex Cover}.

\begin{lemma}\label{t_7b1}
Let $G$ be a graph having no isolating vertices.
Any minimum vertex cover $S_{vc}$ in $G$ is also an optimal solution to \textsc{Weighted Mixed Domination} with $w_e\geq 2w_v$ in $G$.
\end{lemma}
\begin{proof}
 Let $S_{wmd}$ be an optimal solution to \textsc{Weighted Mixed Domination}. The vertex set $V(S_{wmd})$ is still a mixed dominating set by Lemmas~\ref{t_4} and~\ref{t_3}. It holds that
$w(V(S_{wmd}))=w_v|V(S_{wmd})|\leq w_v(|S_{wmd}\cap V|+2|S_{wmd}\cap E|)\leq w_v|S_{wmd}\cap V|+w_e|S_{wmd}\cap E|=w(S_{wmd})$.
Then, $V(S_{wmd})$ is also an optimal solution to \textsc{Weighted Mixed Domination}.
A minimum vertex cover $S_{vc}$ is a mixed dominating set by Lemma~\ref{t_3}.
Note that $V(S_{wmd})$ is a vertex cover by Lemma~\ref{t_4} and then $w(S_{vc})\leq w(V(S_{wmd}))$.
Thus, $S_{vc}$ is an optimal solution to \textsc{Weighted Mixed Domination}.
\end{proof}

\begin{lemma}\label{t_7}
For a weighted instance $G$ having no isolating vertices, if it holds that $w_e\geq 2w_v$,
then any $\alpha$-approximation solution to \textsc{Vertex Cover} is also  an $\alpha$-approximation solution
to \textsc{Weighted Mixed Domination} in $G$.
\end{lemma}
\begin{proof}
Let $C$ be an $\alpha$-approximation solution to \textsc{Vertex Cover}.
The set $C$ is a vertex cover and then it is a mixed dominating set by Lemma~\ref{t_3}.
Next, we consider $w(C)$.
 Let $S_{wmd}$ and $S_{vc}$ be an optimal solution to \textsc{Weighted Mixed Domination} and \textsc{Vertex Cover}, respectively.
Since $|C|\leq \alpha |S_{vc}|$, we have that $w(C)\leq \alpha w(S_{vc})$. By Lemma~\ref{t_7b1}, we have that $w(S_{vc}) = w(S_{wmd})$. Thus, $w(C)\leq \alpha w(S_{wmd})$ and $C$ is also an $\alpha$-approximation solution
to \textsc{Weighted Mixed Domination}.
\end{proof}

The best known approximation ratio for \textsc{Vertex Cover} is 2. Theorem~\ref{t_7} implies that

\begin{theorem}\label{th_2}
\textsc{Weighted Mixed Domination} with $2w_v \leq w_e$ allows polynomial-time 2-approximation algorithms.
\end{theorem}

For lower bounds, we show a reduction from another direction.

\begin{lemma}\label{l_7}
Let $G$ be an instance having no isolating vertices, where $w_e\geq 2w_v$.
For any $\alpha$-approximation solution $D$ to \textsc{Weighted Mixed Domination} in $G$, the vertex set $V(D)$ is an
$\alpha$-approximation solution to \textsc{Vertex Cover} in $G$.
\end{lemma}
\begin{proof}
Let $S_{wmd}$ and $S_{vc}$ be an optimal solution to \textsc{Weighted Mixed Domination} and \textsc{Vertex Cover}, respectively.
By Lemma~\ref{t_7b1}, we have that $w(S_{wmd})=w(S_{vc})$.
Then $w(D)\leq \alpha w(S_{wmd})= \alpha w(S_{vc})=\alpha w_v|S_{vc}|$. Note that $w(D)=w_v|D_v|+w_e|D_e|\geq w_v|D_v|+2w_v|D_e|$
and $|V(D)|\leq  |D_v|+2|D_e|$. Thus, $|V(D)|\leq \alpha |S_{vc}|$. Furthermore, $V(D)$ is a vertex cover by Lemma~\ref{t_4}.
We know that $V(D)$ is an $\alpha$-approximation solution to \textsc{Vertex Cover}.
\end{proof}

Dinur and Safra~\cite{DinurS02} proved that it is $NP$-hard to approximate \textsc{Vertex Cover} within any
factor smaller than $10\sqrt{5}-21$.
Khot and Regev~\cite{KhotR03j} also prove that \textsc{Vertex Cover} cannot be
approximated to within $2-\varepsilon$ for any $\varepsilon>0$ under UGC.
Those results and Lemma~\ref{l_7} imply

\begin{theorem}\label{th_3}
For any $\varepsilon>0$, \textsc{Weighted Mixed Domination} with $2w_v \leq w_e$
 is not
$(10\sqrt{5}-21-\varepsilon)$-approximable in polynomial time unless
$P= NP$, and not $(2-\varepsilon)$-approximable in polynomial time under UGC.
\end{theorem}

\subsection{The case that $w_v\leq w_e< 2w_v$}
To simplify the arguments, in this section, we always assume the initial graph has no degree-0 vertices.
Note that we can include all degree-0 vertices to the solution set directly according to Lemma~\ref{t_2}, which will not
affect our upper and lower bounds.
\subsubsection{Upper bounds}
We show that this case also allows polynomial-time 2-approximation algorithms. Our algorithm is based on a linear programming model
for \textsc{Vertex Cover}.
Note that we are not going to build a linear programming for our problem \textsc{Weighted Mixed Domination} directly.
Instead, we use a linear programming for \textsc{Vertex Cover}.

Linear programming is a powerful tool to design approximation algorithms for \textsc{Vertex Cover} and many other problems.
Lemma~\ref{t_5} and Theorem~\ref{t_6} reveal some connections between \textsc{Weighted Mixed Domination} and \textsc{Vertex Cover}.
 Inspired by these, we investigate approximation algorithms for \textsc{Weighted Mixed Domination} starting from a linear programming model for \textsc{Vertex Cover}.
For a graph $G=(V,E)$, we assign a variable $x_v\in \{0,1\}$ for each vertex $v\in V$ to denote whether it is in the solution set. We can use the following integer programming model (IPVC) to solve \textsc{Vertex Cover}:
\[\begin{array}{*{20}{l}}
\min&\sum_{v \in V} x_v\\
\rm{s.t.} &x_u + x_v \geq 1, \forall uv\in E\\
{}&x_v \in \{0,1\}, \forall v \in V.
\end{array}\]

If relax the binary variable $x_v$ to $0\leq x_v\leq 1$, we get a linear relaxation for \textsc{Vertex Cover}, called LPVC.
We will use $\mathcal{X}'=\{x'_v| v\in V\}$ to denote a feasible solution to LPVC and $w(\mathcal{X}')$ to denote the objective value under $\mathcal{X}'$ on the graph $G$.
LPVC can be solved in polynomial time. However, a feasible solution $\mathcal{X}'$ to LPVC may not be corresponding to a feasible solution to \textsc{Vertex Cover} since the values in $\mathcal{X}'$ may not be integers.
A feasible solution $\mathcal{X}'$ to LPVC is \emph{half integral} if $x'_v \in \{0,\frac{1}{2},1\}$ for all $x'_v \in  \mathcal{X}'$.
Nemhauser and Trotter~\cite{Nemhauser1974Properties} proved some important properties for LPVC.

\begin{theorem} \emph{\cite{Nemhauser1974Properties}}\label{t_8}
 Any basic feasible solution $\mathcal{X}'$ to LPVC is half integral. A half-integral optimal solution to LPVC can be computed in polynomial time.
\end{theorem}

We use $\mathcal{X}^*=\{x^*_v| v\in V\}$ to denote a half-integral optimal solution to LPVC. We partition the vertex set $V$ into three parts $V_1$, $V_{\frac{1}{2}}$ and $V_0$ according to $\mathcal{X}^*$, which are the sets of vertices with the corresponding value $x^*_v$ being $1$, $\frac{1}{2}$ and $0$, respectively.
There are several properties for the half-integral optimal solution.

\begin{lemma} \emph{\cite{Nemhauser1974Properties}}\label{nt_1}
For a half-integral optimal solution, all neighbors of a vertex in $V_0$ are in $V_1$, and there is a matching of size $|V_1|$ between $V_0$ and $V_1$.
\end{lemma}
Lemma~\ref{nt_1} implies that $(V_0,V_1,V_{\frac{1}{2}})$ is a \emph{crown decomposition} (see~\cite{A:crown2} for the definition) and a half-integral optimal solution can be used to construct a 2-approximation solution and a $2k$-vertex kernel for \textsc{Vertex Cover}.

\begin{lemma}\label{vc-p}
For a half-integral optimal solution $\mathcal{X}$ to LPVC, we use $G_{\frac{1}{2}}$ to denote the subgraph induced by 
The size of a minimum vertex cover in $G_{\frac{1}{2}}$ is at least $|V_{\frac{1}{2}}|-m$, where
$m$ is the size of a maximum matching in $G_{\frac{1}{2}}$.
\end{lemma}

\begin{proof}
Let $M$ be a maximum matching in $G_{\frac{1}{2}}$, where $|M|=m$. We use $V_M$ to denote the set of vertices appearing in $M$ and
$R= V_{\frac{1}{2}} \setminus V_M$, where $|R|=|V_{\frac{1}{2}}|-2m$.
Let $C$ be a minimum vertex cover in $G_{\frac{1}{2}}$. We assume that $|C|< |V_{\frac{1}{2}}|-m$ and show a contradiction that $\mathcal{X}$ is not optimal under the assumption.

We partition the vertex set $V_{\frac{1}{2}}$ into two parts $C$ and $I=V_{\frac{1}{2}}\setminus C$. Note that $C$ is a
vertex cover and then $I$ is an independent set. Let $R_I=R\setminus C$ and $R_C=R\cap C$.
Since $|C|< |V_{\frac{1}{2}}|-m$ and $C$ contains at least one vertex in
each edge in $M$, we know that $R_I=C\setminus R$ is not an empty.
A path $P$ in $G_{\frac{1}{2}}$ that alternates between edges not in $M$ and edges in $M$ is called an \emph{$M$-alternating path}.
We use $C_1$ (resp., $I_1$) to denote the set of vertices in $C$ (resp., in $I$) that are contained in some $M$-alternating paths beginning at a vertex in $R_I$. Let $C_2=C\setminus C_1$ and $I_2=I\setminus I_1$. 

We show that 
\begin{enumerate}
\item[(i)] $|C_2\cap V_M|\geq |I_2\cap V_M|$;
\item[(ii)] there is no edge between a vertex in $I_1$ and a vertex in $C_2$.
\end{enumerate}

For (i), if $|C_2\cap V_M|< |I_2\cap V_M|$, then there exists an edge $ab\in M$ such that $a\in C_1$ and $b\in I_2$. Note that $a\in C_1$ and then $a$ is the end of an $M$-alternating path $P$ beginning at a vertex in $R_I$. Since $a$ is the endpoint of an edge $ab$ in $M$, we know that the last edge in the path $P$ is not in $M$. Thus, $P$ plus edge $ab$ is another
$M$-alternating path beginning at a vertex in $R_I$ and them $b$ must be in $I_1$ instead of $I_2$, a contradiction. So $|C_2\cap V_M|\geq |I_2\cap V_M|$.

For (ii), if there is an edge between $a\in I_1$ and  $b \in C_2$, we will show a contradiction that $M$ is not a maximum matching. First of all, we have that $a\not\in R_I$ otherwise $ab$ can be added
into $M$ to get a larger matching. So we know that $a$ is the endpoint of an edge in $M$ and this edge is between $I_1$ and $C_1$.
Furthermore, $a$ is the end of an $M$-alternating path $P$ beginning at a vertex in $R_I$ since $a\in C_1$. So we can get an $M$-alternating path $P'$ by adding edge $ab$ at the end of $P$. Note that $P'$ is an $M$-alternating path with the first edge and the last edge not in $M$. Switching the edges in $M$ and edges not in $M$ on the path can yield a matching having one more edge than $M$, which is a contradiction to the maximum of $M$.

By $|C_2\cap V_M|\geq |I_2\cap V_M|$ and $|C|< |V_{\frac{1}{2}}|-m$, we can get that $|I_1|> |C_1|$.
Note that any vertex in $I_1$ is only possible to adjacent to vertices in $C_1$ in $G_{\frac{1}{2}}$.
In the whole graph $G$, the vertex set $V_0$ is an independent set of vertices with neighbors only in $V_1$.
So there is no edge between $V_0$ and $I_1$. We know that $V_0\cup I_1$ is an independent set of vertices with neighbors
only in $V_1\cup C_1$.
Let $\mathcal{X}'=\{x'_v| v\in V\}$, where $x'_v=0$ if $v\in V_0\cup I_1$, $x'_v=1$ if $v\in V_1\cup C_1$ and $x'_v={\frac{1}{2}}$
if $v\in V_{\frac{1}{2}}\setminus (I_1\cup C_1)$. We can see that $\mathcal{X}'$ is a feasible half integral solution to LPVC.
Since $|I_1|> |C_1|$, we know that the objective value of $\mathcal{X}'$ is smaller than the objective value of
$\mathcal{X}$, which is a contradiction to the fact that $\mathcal{X}$ is an optimal half integral solution to LPVC.
\end{proof}

We are ready to describe our algorithm now. Our algorithm is based on a half-integral optimal solution $\mathcal{X}^*$ to LPVC. We first include all vertices in $V_1$ to
the solution set as vertex elements, which will dominate all vertices in $V_0\cup V_1$ and all edges incident on vertices in $V_1$.
Next, we consider the subgraph $G[V_{\frac{1}{2}}]$ induced by $V_{\frac{1}{2}}$. We find a maximum matching $M$
in $G[V_{\frac{1}{2}}]$ and include all edges in $M$ to the solution set as edge elements. Last, for all remaining vertices in $V_{\frac{1}{2}}$ not appearing in $M$, include them to the solution set as vertex elements.
The main steps of the whole algorithm are listed in Algorithm~\ref{a_1}.

\renewcommand{\tablename}{Algorithm}
\begin{table}
\noindent\fbox{
  \parbox{\textwidth}{%
        \begin{enumerate}
        \item Compute a half-integral optimal solution $\mathcal{X}^*$ for the input graph $G$ and let $\{V_1,V_{\frac{1}{2}},V_0\}$ be the vertex partition corresponding to $\mathcal{X}^*$.
        \item Include all vertices in $V_1$ to the solution set as vertex elements and delete $V_0\cup V_1$ from the graph (the remaining graph is the induced graph $G[V_{\frac{1}{2}}]$).
            \item Find a maximum matching $M$ in  $G=G[V_{\frac{1}{2}}]$ and include all edges in $M$ to the solution set as edge elements.
            \item Add all remaining vertices in $V_{\frac{1}{2}} \setminus V(M)$ to the solution set as vertex elements.
        \end{enumerate}
  }%
}\medskip
\caption{The main steps of the 2-approximation algorithm}
\label{a_1}
\end{table}

We prove the correctness of this algorithm.  First, the algorithm can stop in polynomial time, because Step~1 uses polynomial time by Theorem~\ref{t_8} and all other steps can be executed in polynomial time.
Second, we prove that the solution set returned by the algorithm is a mixed dominating set.

All vertices in $V_0\cup V_1$ and all edges incident on vertices in $V_0\cup V_1$ are dominated by vertices in $V_1$
because the graph has no degree-0 vertices and $\mathcal{X}^*$ is a feasible solution to LPVC.
All vertices and edges in $G[V_{\frac{1}{2}}]$ are dominated because all vertices in $V_{\frac{1}{2}}$ are included to the solution set either as vertex elements or as the endpoints of
edge elements.
We get the following lemma.

\begin{lemma}\label{alg_1}
Algorithm~\ref{a_1} runs in polynomial time and returns a mixed dominating set.
\end{lemma}

Last, we consider the approximation ratio.
Lemma~\ref{nt_1} implies that the size of a minimum vertex cover in the induced subgraph $G[V_0\cup V_1]$ is at least $|V_1|$. By Lemma~\ref{vc-p},
we know that the size of a minimum vertex cover in the induced subgraph $G[V_{\frac{1}{2}}]$ is at least $|V_{\frac{1}{2}}|-m$, where
$m$ is the size of a maximum matching in $G_{\frac{1}{2}}$. So the size of a minimum vertex cover of $G$ is at least $|V_1|+|V_{\frac{1}{2}}|-m$, i.e.,
\eqn{1}{ |S_{vc}|\geq |V_1|+|V_{\frac{1}{2}}|-m.}
Let $D$ denote an optimal mixed dominating set in $G$. By Corollary~\ref{c_1}, we have that $|V_D|+2|E_D|\geq |S_{vc}|$.
By this and $2w_v> w_e$, we have that
\eqn{2}{ w(D)=|V_D|w_v+|E_D|w_e>{\frac{w_e}{2}}|V_D| +w_e|E_D| \geq {\frac{w_e}{2}}|S_{vc}|.}

Let $D'$ denote a mixed dominating set returned by Algorithm~\ref{a_1}. We have that
\[\begin{array}{*{20}{l}}
w(D')&= &|V_1|w_v+m w_e +(|V_{\frac{1}{2}}|-2m)w_v&\\
&\leq &(|V_1|+|V_{\frac{1}{2}}|-m) w_e&\mbox{by~$w_v\leq w_e$} \\
& \leq &|S_{vc}| w_e &\mbox{by~\refe{1}} \\
& \leq &2w(D).&\mbox{by~\refe{2}}
\end{array}\]

\begin{theorem}\label{th_2-1}
\textsc{Weighted Mixed Domination} with $w_v\leq w_e< 2w_v$ allows polynomial-time 2-approximation algorithms.
\end{theorem}

\subsubsection{Lower bounds}

In this section, we give lower bounds for \textsc{Weighted Mixed Domination} with $w_v\leq w_e< 2w_v$.
These hardness results are also obtained by a reduction preserving approximation from \textsc{Vertex Cover}.
Lemma~\ref{t_6} shows that an $\alpha$-approximation algorithm for \textsc{Vertex Cover} implies  a $2\alpha$-approximation algorithm for \textsc{Vertex-Favorable Mixed Domination}.
For \textsc{Weighted Mixed Domination} with $w_e\geq 2w_v$, we have improved the expansion from $2\alpha$ to $\alpha$ in Lemma~\ref{l_7}.
For \textsc{Weighted Mixed Domination} with $w_v\leq w_e< 2w_v$, it becomes harder. We will improve the expansion from $2\alpha$ to $2\alpha-1$.

\begin{lemma}\label{Hard-2}
For any $\alpha \geq 1$, if there is a polynomial-time $\alpha$-approximation
algorithm for \textsc{Weighted Mixed Domination} with $w_v\leq w_e< 2w_v$, then there exists a polynomial-time
$(2\alpha-1)$-approximation algorithm for \textsc{Vertex Cover}.
\end{lemma}
\begin{proof}
For each instance $G=(V,E)$ of \textsc{Vertex Cover},
we construct $|V|$ instances $G_i=(V_i,E_i)$ of \textsc{Weighted Mixed Domination} with $w_v\leq w_e< 2w_v$ such that
a $(2\alpha-1)$-approximation solution to $G$ can be found in polynomial time based on an $\alpha$-approximation
solution to each $G_i$.

For each positive integer $1\leq i \leq |V|$, the graph
$G_i=(V_i,E_i)$  is constructed in the same way.
 Informally, $G_i$ contains a star $T$ of $2n+1$ vertices and an auxiliary graph $G'_i$ such that the center vertex $c_0$ of the star $T$  is connected to all vertices in $G'_i$,
where $G'_i$ contains a copy of $G$, an induced matching $M_i$ with size $|M_i|=i$, and a complete
bipartite graph between the vertices of $G$ and the left part of the
induced matching $M_i$. This is to say, $V_i=V\cup\{a_j\}_{j=1}^i \cup\{b_j\}_{j=1}^i \cup\{c_j\}_{j=0}^{2n}$
and $E_i=E\cup M_i\cup H_i \cup F_i$, where $M_i=\{a_jb_j\}_{j=1}^i$,
$H_i=\{va_j| v\in V, j\in \{1,\dots,i\} \}$, and
$F_i=\{c_0u| u\in V_i\setminus\{c_0\}\}$.
We give an illustration of the construction of~$G_i$ for $i=3$ in Figure~\ref{reduction1}.
In the graphs $G_i$, the values of $w_v$ and $w_e$ can be any values satisfying $w_v\leq w_e< 2w_v$.

\begin{figure}
  \centering
  \includegraphics[width=0.7\textwidth]{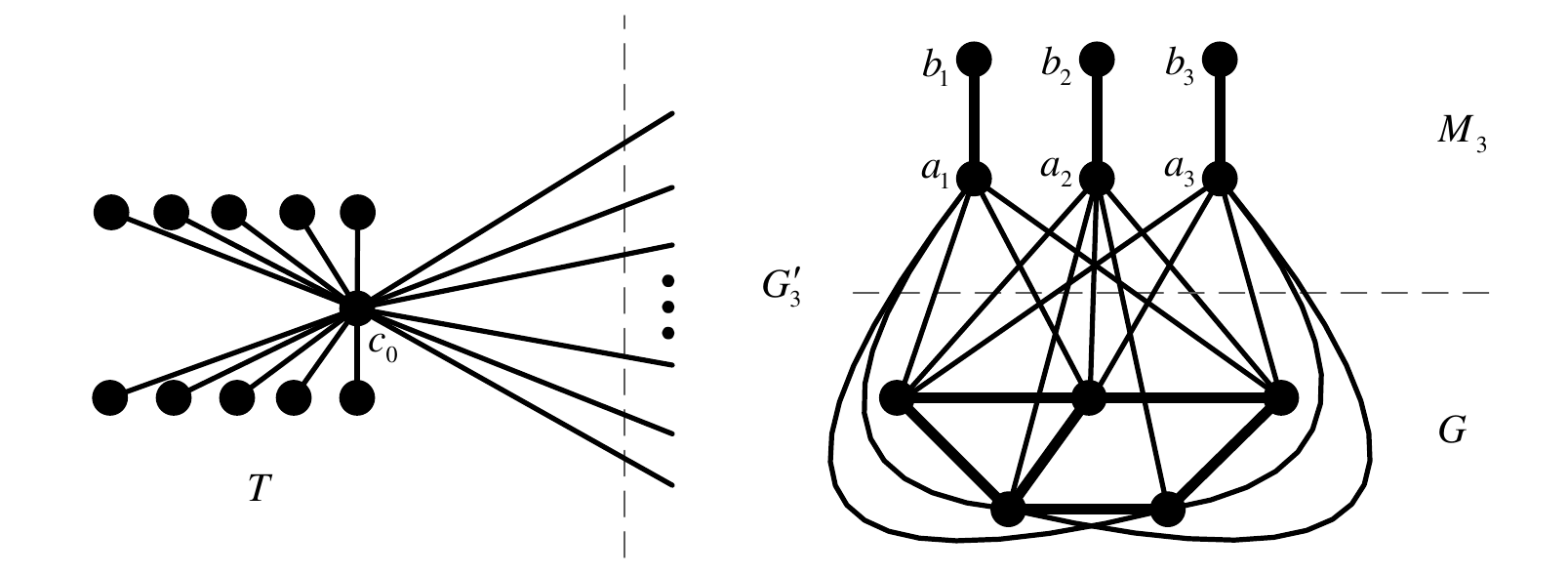}
  \caption{An illustration of the construction of $G_3$}\label{reduction1}
\end{figure}

Let $\tau$ be the size of a minimum vertex cover of $G$. We first show that we can get a $(2\alpha-1)$-approximation solution to $G$ in polynomial time based on an $\alpha$-approximation solution to $G_\tau$.

We define a function $w^*(G')$ on subgraphs $G'$ of $G$ as follows. For a subgraph $G'$ of $G$,
$$w^*(G')=\min_{D\in \mathcal{D}} \{w_v |V(G')\cap V_D|+{\frac{1}{2}}w_e |V(G')\cap V(E_D)|\}.$$
It is easy to see that
\begin{lemma}\label{ratio-0}
Let $S_{wmd}$ be an optimal solution to \textsc{Weighted Mixed Domination} on $G$. It holds that
$$ w(S_{wmd})\geq w^*(G),$$
and for any subgraph $G'$ of $G$ and any subgraph $G_1$ of $G'$, it holds that
$$w^*(G')\geq w^*(G_1)+w^*(G'-V(G_1)).$$
\end{lemma}
Let $D_\tau$ be an optimal solution to $G_\tau$ and $S_{vc}$ be a minimum vertex cover of $G$. By Lemma~\ref{ratio-0} and the definition of the function $w^*()$, we know that
$$w(D_\tau)  \geq w^*(G_\tau)\geq w^*(T)+w^*(G'_\tau).$$
Note that $T$ is a star and then $w^*(T)= w_v$. For $G'_\tau$, we know that the size of a minimum vertex cover of it is at least $2\tau$ because $M_\tau$ is an induced matching of size $\tau$ that needs at least $\tau$ vertices to cover all edges and the size of a minimum vertex cover of $G$ is $\tau$.
By Lemma~\ref{t_4} and $w_e< 2w_v$, we know that $w^*(G'_\tau)\geq \tau w_e$. Thus,
$w(D_\tau)\geq w_v+\tau w_e$.

On the other hand, $D'_\tau=\{c_0\}\cup M'$ is a mixed dominating set with $w(D'_\tau)= w_v+\tau w_e$, where $M'$ is a
perfect matching between $S_{vc}$ and $\{a_j\}_{j=1}^\tau$ with size $|M'|=\tau$. So we have
$$w(D_\tau)= w_v+\tau w_e. $$

Let $D^*_\tau$ be an $\alpha$-approximation solution to $G_\tau$.
We consider two cases. Case~1: the vertex $c_0$ is not a vertex element in $D^*_\tau$.
We will show that the whole vertex set $V$ of $G$ is of size at most $(2\alpha-1)\tau$,
which implies that the whole vertex set is a $(2\alpha-1)$-approximation solution to $G$.
For all the degree-1 vertices $\{c_j\}_{j=1}^{2n}$ in $G_\tau$,
Since all the degree-1 vertices $\{c_j\}_{j=1}^{2n}$ in $G_\tau$ should be dominated and their only neighbor $c_0$ is not a vertex element in the mixed dominating set, we know that $\{c_j\}_{j=1}^{2n}\subseteq V(D^*_\tau)\cap V(T)$.
For $G'_\tau$, an induced subgraph of $G_\tau$, the size of a minimum vertex cover of it is at least $2\tau$.
Let $D''_\tau\subseteq D^*_\tau$ be the set of vertices and edges in $G'_\tau$.
By $w_e < 2w_v$, we know that $w(D''_\tau)\geq \tau w_e$. Thus,
$$w(D^*_\tau) \geq 2n w_v + \tau w_e> (n+ \tau) w_e.$$

On the other hand, we have that
$$w(D^*_\tau)\leq \alpha w(D_\tau)= \alpha (w_v+\tau w_e)\leq\alpha (1+\tau)w_e.$$
Therefore,
$(n+ \tau) w_e < \alpha (1+\tau)w_e$. Thus, $n<\alpha +\alpha \tau -\tau \leq (2\alpha-1)\tau$.

 Case~2: the vertex $c_0$ is a vertex element in $D^*_\tau$.
 For this case, we show that $U_\tau =V(D^*_\tau) \cap V(G)$ is a vertex cover of $G$ with size at most $(2\alpha-1)\tau+ (2\alpha-1)$.
 Since $w(D^*_\tau)\leq \alpha (w_v+\tau w_e)$ and $w_v\leq w_e< 2w_v$, we know that $|V(D^*_\tau)|$ is at most $\alpha (2+2\tau)$.
 Since $M_\tau$ is an induced matching and $T$ is a star, we know that $V(D^*_\tau)$ contains at least $\tau$ vertices in $M_\tau$ and at least one vertex in $T$.
 Therefore,
$$|U_\tau|\leq  \alpha (2+2\tau)-\tau-1=(2\alpha-1)\tau +2\alpha-1.$$
We know that $U_\tau$ is a $(2\alpha-1+\epsilon)$-approximation algorithm for $G$, where $\epsilon= \frac{2\alpha-1}{\tau}$.
In fact, we can also get rid of $\epsilon$ in the above ratio by using one more trick. We let $G'$ be $2\lceil \alpha \rceil$ copies of $G$, and construct
$G_i$ in the same way by taking $G'$ as $G$. The size of the minimum vertex cover of $G'$ is $2\lceil \alpha \rceil \tau$ now. For this case, we will get
$|U_\tau|\leq  (2\alpha-1)2\lceil \alpha \rceil\tau +2\alpha-1$.
Due to the similarity of each copy of $G$ in $G'$, we know that for each copy of $G$ the number of vertices in $U_\tau \cap V(G)$ is at most $(2\alpha-1)\tau +\frac{2\alpha-1}{2\lceil \alpha \rceil}$.
The number of vertices is an integer. So we know that $U_\tau \cap V(G)$ is a vertex cover of $G$ with size at most $(2\alpha-1)\tau$.

\medskip

However, it is $NP$-hard to compute the size $\tau$ of the minimum vertex cover of $G$.
we cannot construct $G_\tau$ in polynomial time directly.
Our idea is to compute $U_i$ for each $G_i$ with $i\in \{1,\cdots, |V(G)|\}$ and return the minimum one $U_{i^*}$.
Therefore, $U_{i^*}$ is a vertex cover of $G$
with size $|U_{i^*}|\leq |U_\tau|$.
\end{proof}

\textsc{Vertex Cover} cannot be approximated within any
factor smaller than $10\sqrt{5}-21$ in polynomial time unless $P=NP$~\cite{DinurS02} and
cannot be approximated within any
factor smaller than $2$ in polynomial time under UGC~\cite{KhotR03j}. These results and Lemma~\ref{Hard-2} imply that

\begin{theorem}\label{hard22}
For any $\varepsilon>0$, \textsc{Weighted Mixed Domination} with $w_v\leq w_e< 2w_v$
 is not
$(5\sqrt{5}-10-\varepsilon)$-approximable in polynomial time unless
$P= NP$, and not $({\frac{3}{2}}-\varepsilon)$-approximable in polynomial time under UGC.
\end{theorem}

\section{\textsc{Edge-Favorable Mixed Domination}}\label{sec:edge}
We show that \textsc{Edge-Favorable Mixed Domination} does not allow polynomial-time constant-ratio approximation algorithms if $P\neq NP$. The hardness result is obtained by a reduction from the \textsc{Set Cover} problem.

In an instance of \textsc{Set Cover}, we are given a set of elements $U=\{1,2,\dots, n\}$ and a collection $\mathcal{S}$
of $m$ nonempty subsets of $U$ whose union equals $U$, and the problem is to find a smallest number of subsets in $\mathcal{S}$
whose union equals $U$.
For an instance $I$ of \textsc{Set Cover}, we construct an instance $I'=(G,w_v,w_e)$ of \textsc{Edge-Favorable Mixed Domination}.
The graph $G=(V=V_S\cup V_U,E)$ is a bipartite graph containing $m+n(q^2+1)$ vertices, where $q=\lfloor m\ln n\rfloor$. The set $V_S$ contains $m$ vertices and
each vertex in $V_S$ is corresponding to a subset in $\mathcal{S}$. The set $V_U$ contains $n(q^2+1)$ vertices in total and
$V_U=V_{1}\cup V_{2}\dots \cup V_{q^2}\cup V_{q^2+1}$, where $|V_i|=n$ and each vertex in $V_i$ is corresponding to an element in $U$
for each $i\in \{1,2,\dots, q^2+1\}$. A vertex $v\in V_S$ is adjacent to a vertex $u\in V_U$ if and only if
the subset corresponding to $v$ contains the element corresponding to $u$.
Thus, if a subset contains $x$ elements, then the corresponding vertex in $V_S$ has degree exactly $x(q^2+1)$.
Let $w_v=1$ and $w_e=\frac{1}{q}$. We first prove the following result.

\textbf{Property 1}: For any ratio $\delta \leq \ln n$, a $\delta$-approximation solution $D^*$ to $I'$ will hold that\\
(i) $V_S \subseteq V(D^*)$, and \\
(ii) the set of subsets corresponding to $V_{D^*}\cap V_S$ is a set cover of $U$.

\medskip

Assume to the contrary that there is a vertex $v \in V_S$ such that $v$ is not in $V(D^*)$. Then all neighbors of $v$ should be in $V(D^*)$. Since $v$ has at least $q^2+1$ neighbors in $V_U$, which are not adjacent to each other, we
know that $D^*$ contains at least $q^2+1$ elements and $w(D^*)\geq w_e(q^2+1)> q$.
Note that the vertex set $V_S$ is a mixed dominating set and then $w(S_{wmd})\leq m$ for an optimal solution $S_{wmd}$ to $I'$.
Therefore, $\frac{w(D^*)}{w(S_{wmd})}> \frac{q}{m} \geq\ln n$, a contradiction.

Also assume to the contrary that the set of subsets corresponding to $V_{D^*}\cap V_S$ is not a set cover of $U$.
Thus there is a vertex $u\in V_U$ such that no neighbor of it is a vertex element in $D^*$,
which implies that $u$ and its $q^2$ twins (vertices in $V_D$ corresponding to the same element in $U$) are in $V(D^*)$.
Therefore, $D^*$ contains at least $q^2+1$ elements and $w(D^*)\geq w_e(q^2+1)> q$. In the same way, we can show a contradiction.
So Property 1 holds.

\medskip

Recall that we use $S_{sc}$ to denote a minimum set cover to $I$ and  $S_{wmd}$ denote an optimal mixed dominating set to $I'$.
We show that
\eqn{eq1}{w(S_{wmd}) =|S_{sc}|+\frac{m-|S_{sc}|}{q}.}

The optimal solution $S_{wmd}$ can be regarded as a 1-approximation solution to $I'$.
By Property 1, we know that $S_{wmd}$ contains at least
$m$ elements in total and at least $|S_{sc}|$ vertex elements. Therefore,
$$w(S_{wmd}) \geq w_v|S_{sc}|+w_e(m-|S_{sc}|)=|S_{sc}|+\frac{m-|S_{sc}|}{q}.$$

Next, we can construct a mixed dominating set $D'$
such that $w(D')= |S_{sc}|+\frac{m-|S_{sc}|}{q}$.
The mixed dominating set $D'$ is constructed as follows:
for each vertex in $V_S$ corresponding to a set in $S_{sc}$, we include it to $D'$ as a vertex element;
for each other vertex in $V_S$, we include an arbitrary edge incident on it to $D'$ as an edge element.
The set $D'$ constructed above is a mixed dominating set because $S_{sc}$ is a set cover (and thus, all vertices in $V_U$ are dominated by vertices in $V_S$) and all vertices in $V_S$ have been included to $D'$ (and thus, all edges will be dominated).
It holds that $w(D')=w_v|S_{sc}|+w_e(m-|S_{sc}|)=|S_{sc}|+\frac{m-|S_{sc}|}{q}$.
Then the optimal value for $I'$ is exactly $|S_{sc}|+\frac{m-|S_{sc}|}{q}$, and \refe{eq1} holds.

Equipped with Property 1 and \refe{eq1}, we are ready to prove the final result.
Let $D^*$ be an $\alpha$-approximation solution to $I'$ and $V_{D^*}$ be the set of vertex elements in $D^*$, where $\alpha \leq \ln n$. We prove that the set $C^*$ of subsets corresponding to $V_{D^*} \cap V_S$ is an $\alpha$-approximation solution to $I$.
By Property 1, we know that $C^*$ is a set cover. Next, we analyze the size of $C^*$.
Since $D^*$ is an $\alpha$-approximation solution to $I'$, we know that $w(D^*)\leq \alpha(|S_{sc}|+\frac{m-|S_{sc}|}{q})\leq
\alpha|S_{sc}|+ \frac{(m-|S_{sc}|)\ln n}{\lfloor m\ln n\rfloor}< \alpha|S_{sc}|+1$.
Thus, $D^*$ contains at most $\alpha|S_{sc}|$ vertex elements and then $|V_{D^*} \cap V_S|\leq \alpha|S_{sc}|$.
So the set $C^*$ of subsets corresponding to $V_{D^*} \cap V_S$ is an $\alpha$-approximation solution to $I$.

\begin{lemma}\label{ratio-3}
For any $\alpha \leq \ln n$, if \textsc{Edge-Favorable Mixed Domination} can be approximated in polynomial time within a factor of $\alpha$, then \textsc{Set Cover} can be approximated in polynomial time within a factor of $\alpha$.
\end{lemma}


It is known that for any $\epsilon> 0$, \textsc{Set Cover} cannot be approximated to $(1-\epsilon)\ln n$ in polynomial time unless $P=NP$~\cite{DS2014}. By this result together with Lemma~\ref{ratio-3}, we get a lower bound for \textsc{Edge-Favorable Mixed Domination}.

\begin{theorem}\label{th_4}
\textsc{Edge-Favorable Mixed Domination} cannot be approximated to $(1-\epsilon)\ln n$ in polynomial time unless $P=NP$,
for any $\epsilon> 0$.
\end{theorem}

\section{Concluding Remarks}\label{sec:con}
Domination problems are important problems in graph theory and graph algorithms.
In this paper, we give several approximation upper and lower bounds on \textsc{Weighted Mixed Domination},
where all vertices have the same weight and all edges have the same weight.
For the general weighted version of \textsc{Mixed Domination} such that each vertex and edge may receive a different weight,
 the hardness results in this paper show that it will be even harder and we may not be easy to get significant upper bounds.
For further study, it will be interesting to reduce the gap between the upper and lower bounds in this paper.

\section*{Acknowledgements}
This work was supported by the National Natural Science Foundation of China,
under grants 61772115 and 61370071.


\end{document}